\documentclass[%
 reprint,
superscriptaddress,
 aps,
 pra,
 twocolumn,
]{revtex4-2}

\usepackage{amsthm,amsmath,amssymb,graphicx,float,subfigure}
\usepackage{graphicx}
\usepackage{bm}
\usepackage{amsthm,amsmath,amssymb,amsfonts,mathrsfs}
\usepackage{mathrsfs}
\usepackage{amsfonts}
\usepackage[colorlinks=true,citecolor=blue,linkcolor=blue,urlcolor=blue,anchorcolor=blue?]{hyperref}

\usepackage{bm}

\usepackage{algorithm}  
\usepackage{algpseudocode}

\newtheorem*{rep@theorem}{\rep@title}
\newcommand{\newreptheorem}[2]{%
\newenvironment{rep#1}[1]{%
 \def\rep@title{#2 \ref{##1}}%
 \begin{rep@theorem}}%
 {\end{rep@theorem}}}
\makeatother

\newtheorem{theorem}{Theorem}

\newtheorem*{assumption*}{Assumption} 

\newtheorem{definition}{Definition}

\begin{document}

\preprint{APS/123-QED}

\title{Simulating Arbitrary Non-Hermitian Dynamics via Quantum Monte Carlo}
\author{Xiaogang Li}
 \affiliation{School of Computer Science, Peking University, Beijing 100871, China}
 \affiliation{Center on Frontiers of Computing Studies, Peking University, Beijing 100871, China}
\author{Kecheng Liu}
\affiliation{School of Computer Science, Peking University, Beijing 100871, China}
\affiliation{Department of Computer Science, Joint Center for Quantum Information and Computer Science, University of Maryland, College Park, USA}
 \author{Qi-Ming Ding}
\email{dqiming94@pku.edu.cn}
 \affiliation{School of Computer Science, Peking University, Beijing 100871, China}
 \affiliation{Center on Frontiers of Computing Studies, Peking University, Beijing 100871, China}
 

\date{\today}
\begin{abstract} 
Simulating the dynamics of non-Hermitian quantum systems is essential for probing the frontiers of modern physics. However, existing quantum algorithms are fundamentally constrained by the probability loss inherent to probabilistic post-selection, extensive quantum circuit depths and the substantial overhead of ancillary qubits. Herein, we introduce a hybrid quantum-classical algorithm grounded in Quantum Monte Carlo (QMC) to concurrently address both limitations. By reformulating non-unitary evolution as statistical sampling of Hamiltonian dynamics, our framework alleviates the probability loss issue induced by post-selection, achieves a reduction in quantum circuit depth, and requires at most one additional ancillary qubit. {This general and flexible approach is capable of simulating arbitrary, time-dependent, and even non-diagonalizable Hamiltonians, providing a unified solution for diverse problems including parity-time $\mathcal{PT}$-symmetric systems, open quantum systems, and non–completely positive trace-preserving (non-CPTP) maps.} {We demonstrate our method on a many-body open system, achieving excellent agreement with exact solutions.} {This work provides a resource-efficient pathway for exploring novel phenomena in complex non-Hermitian systems on near-term quantum devices.}
\end{abstract}
\maketitle
\textbf{\emph{Introduction.---}}\label{introduction}{The study of non-Hermitian systems has unveiled a rich landscape of physical phenomena absent in their Hermitian counterparts~\cite{Bender2007a,el2018non,Kawabata2019,Ashida2020}. These include $\mathcal{PT}$ symmetry~\cite{Bender1998,Bender1999,Bender2024}, exceptional points~\cite{Heiss2004,Minganti2019}, and the non-Hermitian skin effect~\cite{MartinezAlvarez2018,Yao2018,Okuma2020,li2024observation,PhysRevB.97.121401,Foa_Torres_2020}. These features not only challenge the fundamental understanding of quantum mechanics but also open avenues for applications in topological photonics~\cite{Zhao2019,Wang2021,Nasari2023,Yan2023}, dissipative quantum computation~\cite{Verstraete2009,Lin2022}, and quantum sensing~\cite{Chu2020,Ding2023,Xiao2024a}. At the dynamical level, however, the performance and functionality of these platforms are determined by the time evolution generated by effective non-Hermitian Hamiltonians in open quantum settings, which governs stability thresholds, controllability, transient amplification, and metrological sensitivity~\cite{Ashida2020,Zhao2019,Lin2022,Chu2020}. Therefore, a detailed and predictive understanding of these systems requires an efficient simulation method for non-Hermitian dynamics.}

{Despite this importance, simulating non-Hermitian dynamics, particularly in quantum many-body systems, remains a central challenge.} Classical approaches are hindered by the exponential scaling of the Hilbert space~\cite{Feynman2018,Georgescu2014}. While quantum computers offer a promising alternative~\cite{Lloyd1996}, their native unitary logic is ill-suited to implementing non-Hermitian dynamics (i.e., effectively non-unitary evolution). Prevailing quantum algorithms bridge this gap by embedding non-Hermitian dynamics into a larger unitary evolution, for example via linear combinations of unitaries (LCU)~\cite{Childs2012,Berry2015} or dilation techniques~\cite{Huang2018,Huang2019,Li2022}. Yet these approaches are inherently probabilistic and rely on post-selection, which severely suppresses the success probability and induces steep resource overheads, often necessitating amplitude-amplification routines such as Oblivious Amplitude Amplification (OAA)~\cite{berry2014exponential,Guerreschi2019}. Furthermore, the implementation of larger unitary evolutions often relies on Trotter-like deterministic algorithms, which require longer quantum circuits, thereby imposing further constraints on the practical utility of quantum algorithms. The above two factors jointly limit the practicality of quantum algorithms in the noisy intermediate-scale quantum (NISQ) devices.

{In this Letter, we introduce a hybrid quantum–classical framework that reduces the circuit depth and alleviates the post-selection problem by leveraging the statistical power of QMC \cite{Troyer2005,Foulkes2001,Carlson2015,Kliesch2011,Yang2021}.} {By directly sampling a unitary evolution operator decomposition of the non-Hermitian propagator, we circumvent explicit post-selection and its associated probability loss. This approach only requires at most one additional ancillary qubit \cite{Yu2025}. Furthermore, it enables the efficient embedding of various effective simulation algorithms as subroutines for Hamiltonian simulation through Quantum Monte Carlo (QMC) methods \cite{Campbell2019,Nakaji2024,Yang2021,Granet2024}.} The resulting algorithm naturally generalizes the quantum imaginary-time evolution (QITE) paradigm~\cite{Huggins2022,Huo2023,Kamakari2022} from time-independent anti-Hermitian Hamiltonian envolutions to arbitary time-dependent non-Hermitian Hamiltonian evolutions, and enabling the simulation of a broader class of physically significant problems, including $\mathcal{PT}$-symmetric systems, general open quantum systems, and non-CPTP maps~\cite{wei2024simulating}. Our approach is highly versatile, in principle applicable to arbitrary time-dependent and even non-diagonalizable non-Hermitian systems, and its modular design allows seamless integration with quantum error mitigation techniques~\cite{Yang2021,Yu2025,endo2021hybrid,RevModPhys.95.045005}. Practically, this delivers clear advantages: we eliminate reliance on deep amplitude-amplification routines, reduce hardware burden by shifting complexity to classical sampling, and retain broad applicability across non-Hermitian models. To demonstrate feasibility, we apply the framework to simulations of open quantum systems, thereby opening a viable route to problems that were previously challenging, such as dynamics in $\mathcal{PT}$-symmetry–broken regimes \cite{Yu2020,Ding2023}.

\textbf{\emph{Framework.---}}\label{dynamics_simulation} {Our central proposal is to reframe the simulation of non-unitary dynamics as a statistical sampling problem, perfectly suited for a hybrid QMC algorithm. The core challenge in simulating a non-Hermitian Hamiltonian $H(t)$ is that its time-evolution operator, $u(t) = \mathcal{T} e^{-i\int_0^t H(s) ds}$, is non-unitary and cannot be implemented directly on a quantum computer. Instead of attempting to deterministically construct this operator, we leverage a key insight: $u(t)$ can be expressed as a continuous linear combination of \textit{unitary} operators. Given a time-dependent non-Hermitian Hamiltonian $H(t)$, which can always be decomposed into $H(t)=H_r(t)-iH_i(t)\label{HrHi}$, where $H_r(t)=[H(t)+H^\dagger(t)]/2$, $H_i(t)=i[H(t)-H^\dagger(t)]/2$. Building on the framework developed in Refs.~\cite{An2023, An2023a}, when $H_i(t)\geqslant0$, we can write}
\begin{align}\label{Cauthy_distribution_trans}
u(t) {= \int_{\mathbb{R}} g(k) U(t,k) dk.}
\end{align}
{where each $U(t,k) = \mathcal{T} e^{-i \int_0^t K_s(k) ds}$ is a unitary evolution generated by a purely Hermitian Hamiltonian $K_s(k) = H_r(s)+k\cdot H_i(s)$. The function $g(k)\equiv\frac{f(k)}{1-ik}$ acts as a complex-valued, unnormalized probability density over the real variable $k$. This representation is the cornerstone of our method, as it transforms the difficult problem of implementing a single non-unitary operator into the tractable task of sampling from a family of unitary ones.} In another case where the Hamiltonian systems do not satisfy the conditions $H_i(t)\geqslant0$, an additional compensation constant $c_p$ is added to $H_i(t)$ such that $H_i(t)+c_p\geqslant0$. At the same time, the cost is that an additional constant $e^{c_pt}$ needs to be multiplied at the end of the right hand side of the above Eq.\eqref{Cauthy_distribution_trans}. More details regarding the kernel function $f(k)$ can be found in Appendix \ref{kernel_functions}

Given an observable $O=\sum_n o_n O_n$, where $o_n\in\mathbb{Z}$, $O_n$ can be selected as any unitary operator for implementation on quantum circuits, and usually $O_n\equiv\sigma_n\in\{I,X,Y,Z\}^{\otimes m}$ is general Pauli operator. In QMC methods, the focus of the quantum simulation task is on ultimately obtaining the mean value of observables $\langle O\rangle(t)\equiv\langle\psi|u^\dagger(t)\cdot O \cdot u(t)|\psi\rangle/\|u(t)|\psi\rangle\|^2$ to a certain precision, rather than on obtaining the final state $u(t)|\psi\rangle/\|u(t)|\psi\rangle\|$. The factor $\|u(t)|\psi\rangle\|$ represents a normalization factor caused by the non-unitary evolution $u(t)$ in the non-Hermitian system. Therefore, the following theorem guarantees the successful implementation of the above quantum simulation task.

\begin{theorem}[Quantum-Classical Monte Carlo Estimator]
\label{theorem_O_QMC}
Given an initial state $|\psi\rangle$, the arbitrary finite-dimensional time-dependent non-Hermitian Hamiltonian $H(t)$,  a new Hamiltonian $\tilde{H}(s)=H_r(s)-i(H_i(s)-{E_i}_0(s))$ with a manually chosen time-dependent quantity ${E_i}_0(s)$, the observable $O=\sum_n o_n O_n$, then the estimation $\langle O\rangle(t)$ of the observable $O$ can be synthesized through unbiased estimation of the numerator part $N(O)\equiv\langle\psi|\mathcal{\overline{T}} e^{-i\int_0^T \tilde{H}^\dagger(s) d s}\cdot O\cdot\mathcal{T} e^{-i\int_0^T \tilde{H}(s) d s}|\psi\rangle$ and the denominator part $D$ (equivalently denoted as $N(I)$) separately,
\begin{align}\label{O_est}
\langle O\rangle=\frac{N(O)}{D}=\frac{\|O\|_{l_1}\cdot\mathbb{E}_{k',k,n}\langle O_n(k',k)\cdot e^{i\theta(n)}\cdot e^{i\theta(k',k)})\rangle}{\mathbb{E}_{k',k}\langle I(k',k)\cdot e^{i\theta(k',k)}\rangle}.
\end{align}
 Here $\langle [\cdot](k',k)\rangle\equiv$ $\langle\psi|\mathcal{\overline{T}} e^{i \int_0^T K_{s'}(k') d s'}[\cdot]\mathcal{T} e^{-i \int_0^T K_s(k) d s}|\psi\rangle$, $\mathcal{\overline{T}}$ is the anti-time-ordering operator, $K_s(k)= H_r(s)-k\cdot (H_i(s)-{E_i}_0(s))$, where the compensation constant ${E_i}_0(s)$ can be chosen as any real value to ensure that $H_i(s)-{E_i}_0(s)\geqslant0$. $\mathbb{E}_{k',k,...}\equiv\mathbb{E}_{k'}\mathbb{E}_k,...\mathbb{E}_{[\cdot]}$ represents the expectation with respect to random variables $k$, $k'$, $...$, the probability density related to $k$ is $\mathrm{Pd}(k)$, $e^{i\theta(k',k)}=e^{i(\theta(k)-\theta(k'))}=g^*(k')g(k)/\|g\|_1^2$ is the phase part, and the sampling probability of $O_n$ is $p_n=|o_n|/\|O\|_{l_1}$, where $\|O\|_{l_1}\equiv\sum_n |o_n|$, $e^{i\theta(n)}\equiv o_n/|o_n|$ denotes the phase part. Noting that $\|O\|_{l_1}\cdot\mathbb{E}_{k',k,n}\langle O_n(k',k)\cdot e^{i\theta(n)}\cdot e^{i\theta(k',k)})\rangle=\mathbb{E}_{k',k}\langle O(k',k)\cdot e^{i\theta(k',k)})\rangle$.
\end{theorem}
\begin{proof}
The proof is provided in Appendix \ref{proof_O_est}.
\end{proof}

It is worth emphasizing that the estimation method described by the Thm. \ref{theorem_O_QMC} is unbiased and employs no approximation techniques, thereby reducing sources of systematic error. This distinguishes it from existing biased estimation approaches \cite{Huo2023}.
{The power of this framework lies in its structure. The classical computer handles all statistical sampling and summation, while the quantum computer is tasked only with what it does best: executing the unitary evolution $U(t,k)$ and measuring the resulting overlaps. Each run on the quantum device, regardless of the specific values of $k$ and $k'$, provides a useful quantum circuit that contributes to the final average. This elegantly sidesteps the post-selection problem inherent in many other approaches, where a single failed ancillary measurement can invalidate an entire computational run. Meanwhile, through sampling, various Hamiltonian simulation subroutines can be conveniently embedded \cite{Yang2021,Campbell2019,Nakaji2024,Granet2024}, and optimized by various error mitigation methods \cite{Yang2021,Yu2025}. } {Our QMC approach therefore alleviates the probability loss due to post-selection, reduces the quantum circuit depth, and reduces the number of additional ancillary qubits to at most one, making it inherently resource-efficient and especially well-suited for near-term quantum hardware.}

\textbf{{\label{complexity_analysis_theory}}\emph{Error and complexity analysis.---}}The algorithm will estimate the numerator (the part of $O_n(k',k)$) and denominator (the part of $I(k',k)$) of Eq.\eqref{O_est} respectively. There are two types of errors here in this algorithm, one is systematic error, and the other is statistical error. The systematic error is mainly caused by the limited resolution of quantum gates since the gate angles are defined and implemented using digital electronics \cite{Koczor2024,Granet2024}, and truncation errors since the domain of the integrand function $g(k)$ is $(-\infty,\infty)$ and thus it has to be truncated within a finite interval \cite{Huo2023}. The statistical error has two contributing factors, the fluctuation error caused by limited samplings, and the shots noise caused by limited measurements. 

\begin{theorem}\label{theorem_O_QMC_error}
Given that $\eta$ is the permissible error synthesized by the numerator part and the denominator part, $\delta$ is the statistical failure probability bound, and $\varepsilon$ is the given truncation error, denoting that $\hat{N}(O)=\|g\|_1^2\cdot\langle O(k',k)\cdot e^{i\theta(k',k)}\rangle, \overline{\hat{N}}(O)=\frac{1}{n_N}\sum_{i=1}^{n_N}\hat{N}(O)$, then the inequality 
\begin{align}\label{O_est_error}
\left|\frac{\overline{\hat{N}}(O)}{\overline{\hat{D}}}-\langle O\rangle\right|\leqslant \eta
\end{align}
hold with a probability higher than $1-\delta$. When Hamiltonian $\tilde{H}$ is time-dependent, the algorithm can be effectively implemented with a duration of $T=\mathcal{O}(1/\Delta_{E_i})$, $\Delta_{E_i}$ represents the maximum bandwidth of $\tilde{H}_i(t)=H_i(s)-{E_i}_0(s)$. The query complexity required to the above inequality is $k_c=\mathcal{O}\left(T \max _t\|\tilde{H}(t)\|\left(\log \left(\frac{1}{\varepsilon}\right)\right)^{1+1 / \beta}\right)$, the maximum sampling complexity is $n_N=n_D=\mathcal{O}\left(\frac{\log(\frac{1}{\delta})\cdot(\|O\|_{l_1}+1)^2\|g\|_1^4}{\eta^2}\right)$. When $\tilde{H}$ is time-independent, there is no restriction on the duration of the algorithm, and the query complexity required is $k_c=\mathcal{O}\left(\left(\log \left(\frac{1}{\varepsilon}\right)\right)^{1 / \beta}\right)$, the maximum sampling complexity is $n_N=n_D=\mathcal{O}\left(\frac{\log(\frac{1}{\delta})\cdot(\|O\|_{l_1}+1)^2\|g\|_1^4}{\eta^2p_g^2}\right)$, where $p_g\equiv |\langle\phi_g|\psi\rangle|^2$ is the non-zero overlap between the ground state $|\phi_g\rangle$ of $\tilde{H}_i$ and the initial state $|\psi\rangle$.

\end{theorem}
 
 \begin{proof}
   The proof of this theorem will be provided in Appendix \ref{complexity_analysis}.
 \end{proof}

Finally, by embedding any known Hamiltonian simulation subroutine, the final simulation can be achieved in Appendix \ref{Hamiltonian_simulation}. For reference, we recall the {Hamiltonian simulation without discretization errors (HSWDE)} \cite{Granet2024}, which also belongs to a class of quantum Monte Carlo methods. It is particularly worth noting that in this subroutine, the limited quantum gate resolution does not contribute to systematic errors. Thus, in the absence of sampling limitations, truncation error stands as the sole source of systematic error in theory. In summary, both systematic and statistical errors can be simultaneously reduced simply by increasing the number of samples. Moreover, this subroutine has been shown to be compatible with time-dependent Hamiltonians as well as non-sparse Hamiltonians, and is particularly well-suited for present-day relatively noisy quantum hardware. The effect of this subroutine can be seen in the Fig.\ref{Loschmidt_amplitude_QMC}.

\textbf{{\label{c_and_a}}\emph{Circuits and Algorithms.---}}The circuit of our algorithm is shown in Fig.\ref{hybrid_algorithm}, and the process of our algorithm for estimating the numerator $N(O)$ and the denominator $D$ are shown in Algorithm \ref{alg:ND}. Their measurements are carried out separately for the real and imaginary parts, as detailed in Appendix \ref{ND_measurements}. From the Fig.\ref{hybrid_algorithm}, we know that it requires at most one additional ancillary qubit.

\textbf{\label{applications}\emph{Applications.}---}Simulating open quantum systems has long been a challenging task, primarily due to the fact that their evolution follows non-unitary dynamics, governed by CPTP maps. Currently, there exist some methods for such simulations \cite{childs2016efficient,cleve2016efficient,li2022simulating,Kamakari2022,Ding2024}. In this section, we solve this problem based on the aforementioned theoretical framework, namely Algorithm \ref{alg:ND} constructed from Eq.\eqref{O_est}. Compared to previous algorithms, our approach does not require frequent measurements or resetting of ancillary qubits in the middle of the procedure \cite{cleve2016efficient,Kamakari2022,Han2021}, nor does it necessitate complex classical computations involving intermediate matrix variables or complex isomorphisms \cite{childs2016efficient,li2022simulating,Ding2024}.


 \begin{figure}[htb]
  \centering
  \includegraphics[width=0.95\linewidth]{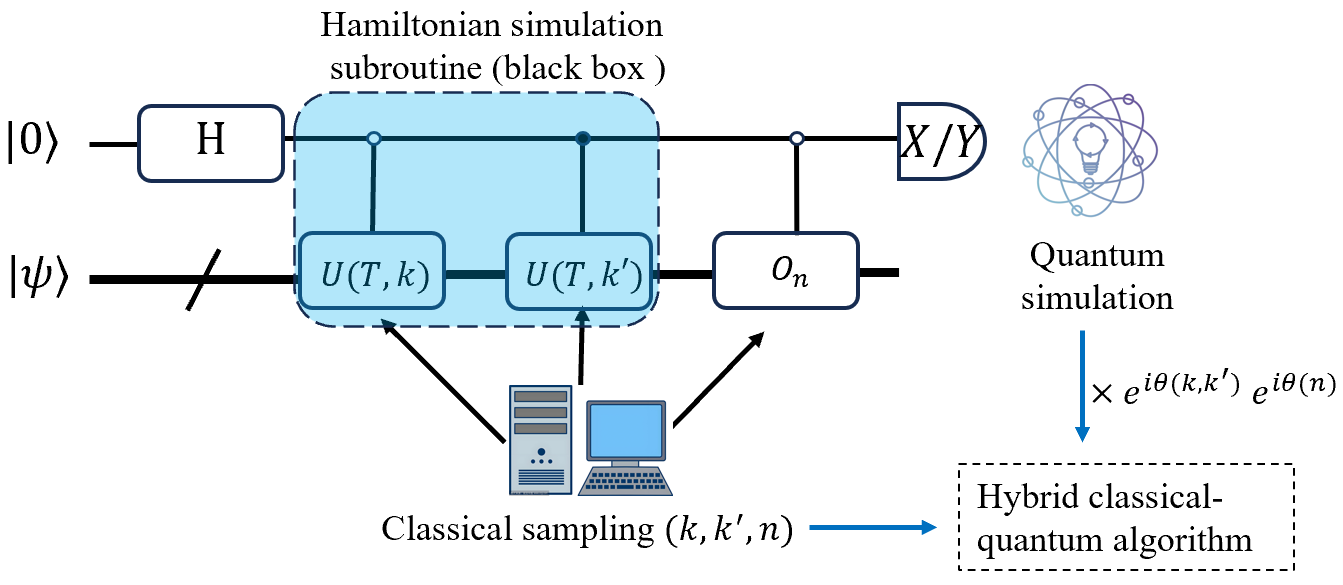}
  \caption{Hybrid classical-quantum algorithm for simulating arbitrary time-dependent non-Hermitian systems based on Quantum Monte Carlo. $X$ ($Y$) denotes the measurement basis, which is related to the real (imaginary) part of estimated value.}\label{hybrid_algorithm}
\end{figure}

The evolution equation of an open quantum system under Born-Markov approximation and rotating wave approximation can be expressed by Lindblad master equation \cite{breuer2002theory}: 
 \begin{equation}\label{master equation}
   \frac{\mathrm{d}\rho(t)}{\mathrm{d} t}=\mathcal{L}\rho(t)=-i[H(t),\rho(t)]+\sum_{\mu}\mathcal{D}[\Gamma_\mu]\rho(t),
 \end{equation}
where $\rho(t)$ is the density operator of the system, $\Gamma_\mu$ is the jump operator, $\mathcal{L}$ is the Liouvillian superoperator, which belongs the CPTP map, and $\mathcal{D}[\Gamma_\mu]$ is the dissipator related to the jump operator $\Gamma_\mu$, which is used to describe the dissipation:
$
  \mathcal{D}\left[\Gamma_\mu\right] \rho(t)=\Gamma_\mu \rho(t) \Gamma_\mu^\dag-\frac{\Gamma_\mu^\dag \Gamma_\mu}{2} \rho(t)-\rho(t) \frac{\Gamma_\mu^\dag \Gamma_\mu}{2}.
$

Though the vectorization method, the state $\rho(t)=\sum_{i,j}\rho_{ij}|i\rangle\langle j|$ will be mapped to a vector $|\rho(t)\rangle\!\rangle=\sum_{i,j}\rho_{ij}|i\rangle|j^*\rangle$, which is an unnormalized quantum state, and can be normalized to a legal quantum state $|\rho(t)\rangle=|\rho(t)\rangle\!\rangle/\||\rho(t)\rangle\!\rangle\|=|\rho(t)\rangle\!\rangle/\|\rho(t)\|_F$, where "F" denotes the Frobenius norm, and $\||\rho(t)\rangle\!\rangle\|=\|\rho(t)\|_F$. Then the superoperator process, such as $A\cdot\rho\cdot B$, can be mapped to $ A\otimes B^\mathrm{T}\cdot |\rho\rangle\!\rangle$. Therefore, based on the vectorization techniques, non-physical processes can also be well characterized (see the Appendix in Ref.\cite{Minganti2019} for details), and when combined with the method we proposed in Eq.\eqref{O_est}, these processes can also be simulated. It is worth noting that  the denominator part happens not to appear in the simulation of open quantum systems. It is worth noting that, we have shown that it renders unnecessary any compensation by a time-exponential constant $e^{c_p t}$ if the jump operators are normal operators, typical examples include bit flip noise (X), phase flip noise (Z), and bit-phase flip noise (Y). The details can be seen in Appendix \ref{open_systems_vec}.

As an example, we consider a 4-qubit transverse Ising model with periodic boundary conditions subjected to a dissipation process of the amplitude damping (AD),
\begin{align}
H=-J\underset{\langle i,j\rangle}\sum Z_iZ_j-h\underset{j}\sum X_j.
\end{align}
For convenience, we assume that the AD noise only acts on the first qubit, and the jump operator of AD noise is $\Gamma=\sqrt{\gamma}|0\rangle_1\langle1|$, where $\gamma=1.5$. This model is shown in Fig.\ref{Ising_model_4qubits}.
 \begin{figure}[htb]
  \centering
  \includegraphics[width=0.45\linewidth]{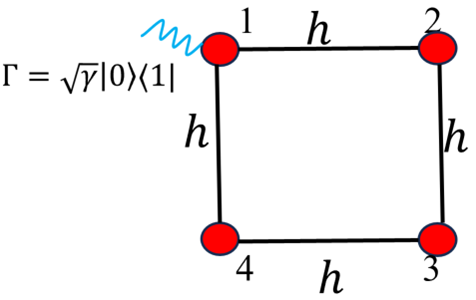}
  \caption{The 4-qubit transverse Ising model with periodic boundary conditions subjected to dissipation (amplitude damping noise). For convenience, the dissipation only acts on the first qubit.}\label{Ising_model_4qubits}
\end{figure}
\begin{figure}[htb]
  \centering
  \includegraphics[width=0.95\linewidth]{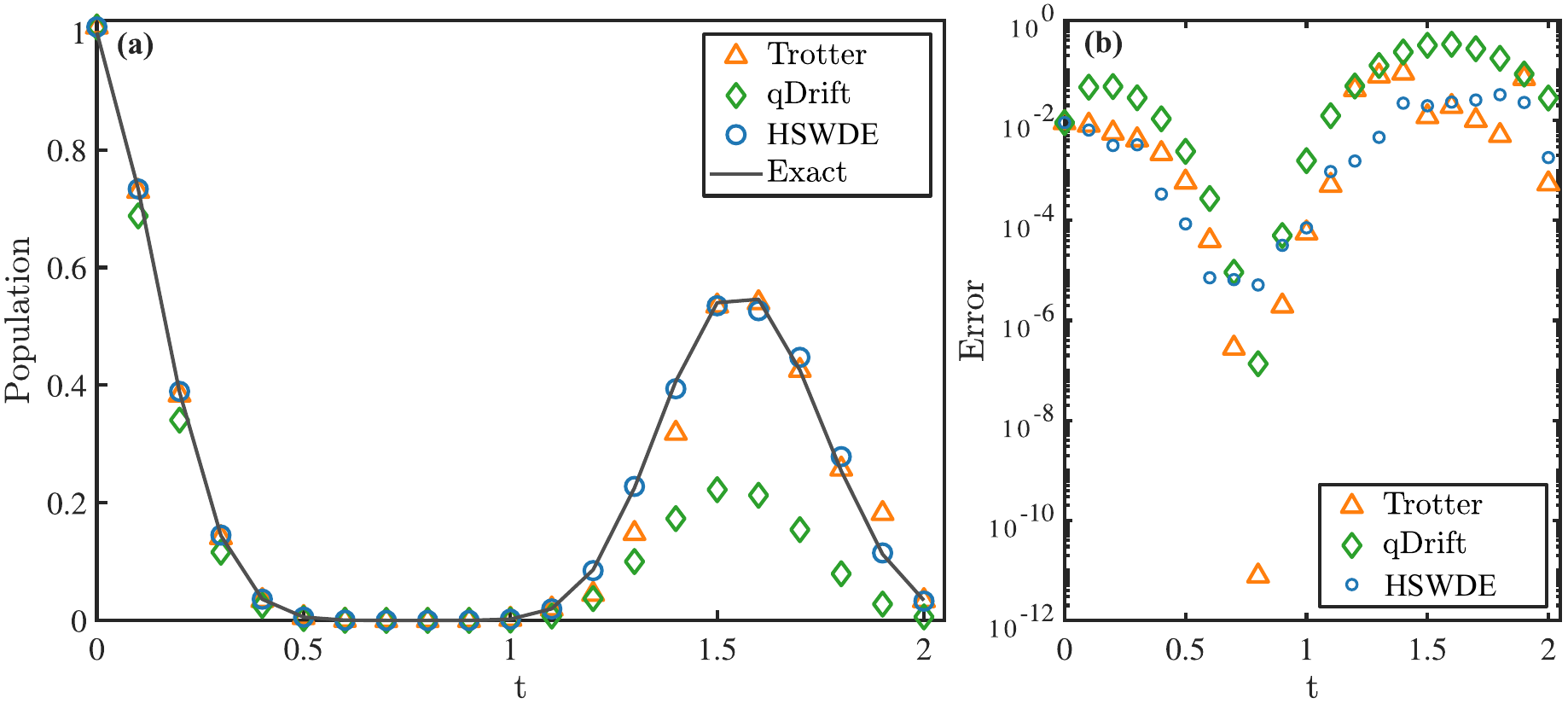}
  \caption{Dynamics simulation of a 4-qubit transverse Ising model with periodic boundary conditions subjected to dissipation. The related parameters in the model are $h=2, J=1$, the initial state $|\psi\rangle=|1000\rangle$, the population observable is $O=|1000\rangle\langle1000|$, and the jump operator of amplitude damping noise $\Gamma=\sqrt{\gamma}|0\rangle_1\langle1|$, where $\gamma=1.5$. (a) The relation between population and time $t$. The simulation duration is $T=2$. The time step of both Trotter, qDrift subroutines is $\Delta t=0.05$, while for HSWDE, the minimum rotation angle (time step) is $\tau=0.05$, the compensation coefficient $c_p=0.3607$, the number of samples $N_s=10^5$ in the HSWDE subroutine. (b) The relation between the absolute error and time $t$. The vertical axis scale is logarithmic.The codes are available
 in \cite{xiaogang_2025_17309084}.}\label{Loschmidt_amplitude_QMC}
\end{figure}

\begin{algorithm}[H]
\caption{\centering{The numerator $N(O)$ (denominator $N(I)$, or $D$) estimation} }
\label{alg:ND}
\begin{algorithmic}[1]
\Require The initial state $|\psi_0\rangle$; the duration time $T$ of simulation; the integrand function $g(k)$; the samples $(k',k)$, so the phase part $e^{i\theta(k',k)}$, $U^\dagger(T,k')$ and $U(T,k)$ given in Eq.\eqref{O_est}; the failure probability bound $\delta$. 
\Ensure  Estimation of the numerator $N(O)$ ($D$) with the failure probability not higher than $\delta$.
\For {$n = 1$ to $n_N$ ($n_D$)}  \Comment{When estimating the denominator $D$, the observable $O$ in $N (O)$ is replaced by $I$.}
    \State Sample two non-negative random numbers $k', k$ separately according to the probability density $\mathrm{Pd}(k)=|g(k)|/\|g(k)\|_{l_1}$.
    \State Sample an unitary observable $O_n$ according to the sampling probability $p_n=|o_n|/\|O\|_{l_1}$.
    \State Apply $U(T,k)$, $O_n$, and $U^\dagger(T,k')$ in order to the initial state $|\psi\rangle$.
    \State Measure the above quantum circuit  by separating into real ($X$ measurement) and imaginary parts ($Y$ measurement) according to Eq.\eqref{XY_measurement}, and record the measurement result.
    \State Multiply by the sampled phase $e^{i\theta(k,k')}$, $e^{i\theta(n)}$ and the constant $\|g\|_1^2\|O\|_{l_1}$, and record it as $\hat{N}(O)$.
\EndFor
\State Calculate the estimator $\overline{\hat{N}}(O)$ as an estimation of $N(O)$ according to the Eq.\eqref{n_N_average}.
\end{algorithmic}
\end{algorithm}

In Fig.\ref{Loschmidt_amplitude_QMC}, we give the simulation results.  The model under investigation incorporates specific parameters detailed as follows: $h = 2$, $J = 1$, with an initial state $|\psi\rangle = |1000\rangle$. The population observable is defined as $O = |1000\rangle\langle1000|$. When implementing specific sampling $U(T,k)$, we used three methods, 1-order Trotter (red triangle) \cite{Yang2021}, qDrift (green diamond) \cite{Campbell2019,Nakaji2024}, and HSWDE subroutine (blue circle) \cite{Granet2024} respectively. The relation between population and time $t$ is given in Fig.\ref{Loschmidt_amplitude_QMC}(a).

As seen in Fig.\ref{Loschmidt_amplitude_QMC}(a), all three methods perform well before the simulation duration of $t=1$. However, beyond this point, the qDrift method deviates significantly from the exact values, while the other two methods exhibit better performance. This phenomenon aligns with the results presented in the Fig.3 of Ref.\cite{Granet2024}. Moreover, it is evident that, overall, within the simulation process, the HSWDE subroutine generally performs the best, only with few exceptions. This is clearly illustrated in Fig.\ref{Loschmidt_amplitude_QMC}(b), which depicts the relation between error and time $t$. The smaller errors imply that, for a specific precision, we can utilize shallower quantum circuits. This is the primary reason we have emphasized the HSWDE subroutine in the preceding text. It is important to note that, unlike the other two methods, the HSWDE subroutine does not suffer from systematic errors due to time discretization. More accurately, the errors arising from time discretization in this method can be automatically alleviated by simply increasing the number of samples, thus avoiding the impact of limited resolution due to finite qubit numbers in practical quantum gates. In contrast, the other two methods inevitably suffer from systematic errors due to time discretization. It should be noted that the unusually small error observed for the Trotter method near $t=0.8$ is a fortuitous outcome.

\textbf{\label{conclusions}\emph{Conclusions.}---}In {summary}, we have {introduced} a general hybrid {quantum-classical} algorithmic framework based on QMC for simulating the dynamics of arbitrary time-dependent non-Hermitian systems, {along with} its error and complexity analyses. {Our framework recasts non-unitary evolution as a statistical sampling of unitary evolutions. This strategy avoids the post-selection bottleneck common to many quantum algorithms, reduces circuit depth, and uses at most one additional ancillary qubit. Consequently, the method is inherently resource-efficient and well-adapted to the constraints of near-term hardware.} This algorithm can be applied to simulate any time-dependent non-Hermitian system dynamics, including PT-symmetric systems, {non-CPTP maps}, and open quantum systems, without assuming that the Hamiltonian is {of a special form (e.g.,} anti-Hermitian or diagonalizable{)}, {as required by methods like} the quantum imaginary time evolution algorithm, {of which our work can be seen as a natural generalization}. This algorithm combines the advantages of both classical and quantum computation, requiring {minimal} ancillary qubits {while accommodating} shallow quantum circuits. To verify the effectiveness of our method, we applied it to the task of simulating {a multi-qubit} open quantum system and achieved {results in excellent agreement with exact solutions}.

{Furthermore, the versatility} of our method {opens several avenues for future research and application.} We anticipate that it will find {value} in more applications, such as solving ground state energies and simulating {other complex} processes. From an optimization perspective, {our work also provides practical guidance for implementation}. When performing real simulations on quantum computers, Hamiltonian simulation needs to be used as a subroutine within our framework. In our application, we used the 1st-order Trotter method, the qDrift method, and the HSWDE method as subroutines, and found that {the HSWDE method generally offers superior accuracy, particularly when the number of samples is moderate.} {This advantage stems from its ability to systematically reduce both systematic and statistical errors by simply increasing the sample size, a feature not shared by the other methods tested}. Therefore, we recommend the HSWDE method as the preferred approach for our Hamiltonian simulation subroutine. A natural next step would be to incorporate various quantum error mitigation techniques {into our framework} to achieve even shallower quantum circuit depths {and enhance performance on noisy quantum devices}.

\textbf{\label{Acknowledgements}\emph{Acknowledgements.}---}This work is supported by Beijing Natural Science Foundation (Grant No. 1254053). the National Natural Science Foundation of China (Grant No.~12361161602 and No.~12175003),  NSAF (Grant No.~U2330201), the Innovation Program for Quantum Science and Technology (Grant No.~2023ZD0300200).

\bibliography{NHS_QMC.bib}
\clearpage
\onecolumngrid
\appendix

\section{The kernel functions  \label{kernel_functions}}

Initially, Dong An et al.~\cite{An2023, An2023a} found and strictly proved that the kernel function can be taken as
 \begin{align}\label{Cauthy_distribution}
  f(k)=\frac{1}{\pi(1+ik)}, 
 \end{align}
 then the integrand function $g(k)=\frac{1}{\pi(1+k^2)}$ is the standard Cauchy distribution function, and can be generated effectively and directly by the inverse transform sampling method. However, this integrand has the property: $\mathbb{E}k=0, \mathbb{E}|k|=\infty, \mathbb{E}k^2=\infty$, so in actual use, it is usually necessary to truncate it. Given a truncation precision $\varepsilon$, we can solve the defined truncation length $k_c$ in the following way:
 \begin{align}
 \int_{-k_c}^{k_c} \left|\frac{f(k)}{1-ik}\right|dk=1-\varepsilon.
 \end{align}
 For the kernel function given in Eq.\eqref{Cauthy_distribution}, $k_c=\frac{1}{\tan(\frac{\pi}{2}\varepsilon)}=\mathcal{O}(\frac{1}{\varepsilon})$. 
 
 Later, Dong An et al. put forward a set of improved kernel functions to make the integrand converge faster,
 \begin{align}
 f(k)=\frac{2}{2\pi e^{-2\beta}e^{(1+ik)^\beta}}, \beta \in (0,1).
 \end{align}
This new kernel function decays at a near-exponential rate of $e^{-c|k|^\beta}$, so the truncation length can be reduced from $\mathcal{O}(1/\varepsilon)$ to $\mathcal{O}\left((\log(1/\varepsilon))^{1/\beta}\right)$. This important improvement can further reduce the circuit depth of our algorithm. However, it should be noted that in Eq.\eqref{Cauthy_distribution}, $g(k)$ is real, while it is complex here, which may increase the difficulty of classical sampling. Some important classical Monte Carlo algorithms can be suggested to generate such complex samples, such as the acceptance-rejection sampling algorithm or the Metropolis algorithm, or more directly, use trapezoidal rules to generate such complex probability distributions, but this may introduce additional discretization errors.
\section{\label{proof_O_est}Proof of Theorem \ref{theorem_O_QMC}}
\begin{proof}
Considering the Eq.\eqref{Cauthy_distribution_trans}, we know that
\begin{align}\label{O_est_proof}
\langle O\rangle=&\frac{\langle\psi|\mathcal{\overline{T}} e^{-i\int_0^T H^\dagger(s) d s}\cdot O\cdot\mathcal{T} e^{-i\int_0^T H(s) d s}|\psi\rangle }{\langle\psi|\mathcal{\overline{T}} e^{-i\int_0^T H^\dagger(s) d s}\cdot I\cdot\mathcal{T} e^{-i\int_0^T H(s) d s}|\psi\rangle } \nonumber\\
=&\frac{\langle\psi|\left(e^{\int_{0}^{T}{E_i}_0(s)ds}\cdot\mathcal{\overline{T}} e^{-i\int_0^T H^\dagger(s) d s}\right)\cdot O\cdot \left(e^{\int_{0}^{T}{E_i}_0(s)ds}\cdot\mathcal{T} e^{-i\int_0^T H(s) d s}\right)|\psi\rangle }{\langle\psi|\left(e^{\int_{0}^{T}{E_i}_0(s)ds}\cdot\mathcal{\overline{T}} e^{-i\int_0^T H^\dagger(s) d s}\right)\cdot I\cdot \left(e^{\int_{0}^{T}{E_i}_0(s)ds}\cdot\mathcal{T} e^{-i\int_0^T H(s) d s}\right)|\psi\rangle } \nonumber\\
=&\frac{\langle\psi|\mathcal{\overline{T}} e^{-i\int_0^T \tilde{H}^\dagger(s) d s}\cdot O\cdot\mathcal{T} e^{-i\int_0^T \tilde{H}(s) d s}|\psi\rangle }{\langle\psi|\mathcal{\overline{T}} e^{-i\int_0^T \tilde{H}^\dagger(s) d s}\cdot I\cdot\mathcal{T} e^{-i\int_0^T \tilde{H}(s) d s}|\psi\rangle } \nonumber\\
=&\frac{\langle\psi|\int_{\mathbb{R}} g^*(k') [\mathcal{\overline{T}} e^{i \int_0^T K_{s'}(k') d s'} ]dk'\cdot O\cdot \int_{\mathbb{R}} g(k) [\mathcal{T} e^{-i \int_0^T K_s(k) d s} ]dk|\psi\rangle }{\langle\psi|\int_{\mathbb{R}} g^*(k') [\mathcal{\overline{T}} e^{i \int_0^T K_{s'}(k') d s'} ]dk'\cdot I\cdot \int_{\mathbb{R}} g(k) [\mathcal{T} e^{-i \int_0^T K_s(k) d s} ]dk|\psi\rangle } \nonumber\\
=&\frac{\|O\|_{l_1}\cdot\mathbb{E}_n\int_{\mathbb{R}}\int_{\mathbb{R}} g^*(k')g(k)\cdot \langle O_n(k',k) \cdot e^{i\theta(n)}\rangle\cdot dk'dk}{\int_{\mathbb{R}}\int_{\mathbb{R}} g^*(k')g(k)\cdot \langle I(k',k)\rangle\cdot dk'dk} \nonumber\\
=&\frac{\|O\|_{l_1}\cdot\|g\|_1^2\cdot\mathbb{E}_{n,k',k}\langle O_n(k',k)\cdot e^{i\theta(n)}\cdot e^{i\theta(k',k)}\rangle}{\|g\|_1^2\cdot\mathbb{E}_{k',k}\langle I(k',k)\cdot e^{i\theta(k',k)}\rangle} \nonumber\\
=&\frac{\|O\|_{l_1}\cdot\mathbb{E}_{k',k,n}\langle O_n(k',k)\cdot e^{i\theta(n)}\cdot e^{i\theta(k',k)}\rangle}{\mathbb{E}_{k',k}\langle I(k',k)\cdot e^{i\theta(k',k)}\rangle}\\= &\frac{\mathbb{E}_{k',k}\langle O(k',k)\cdot e^{i\theta(k',k)})\rangle}{\mathbb{E}_{k',k}\langle I(k',k)\cdot e^{i\theta(k',k)}\rangle}.
\end{align}

Noting that when $\tilde{H}$ is time-independent, $\langle [\cdot](k',k)\rangle=  \langle\psi|e^{i K(k')T}[\cdot]e^{-i K(k)T}|\psi\rangle$.
\end{proof}

\section{\label{complexity_analysis} Proof of Theorem \ref{theorem_O_QMC_error} (Error and complexity analysis)}
\begin{proof}
In the following, we present the error and complexity analysis, thus completing the proof of Theorem \ref{theorem_O_QMC_error}.
\subsubsection{Truncation error, quadrature error and the query complexity}
According to the conclusions of Lemma 9 in Ref.\cite{An2023a}, we know that the truncation error incurred in approximating the non-unitary evolution $\mathcal{T} e^{-i\int_0^T \tilde{H}(s) d s}$ in theorem \ref{theorem_O_QMC_error} can be bounded as
\begin{align}
\left\|\int_{\mathbb{R}} g(k) U(T, k) \mathrm{d} k-\int_{-k_c}^{k_c} g(k) U(T, k) \mathrm{d} k\right\| \leq \frac{2^{\lceil 1 / \beta\rceil+1}\lceil 1 / \beta\rceil!}{C_\beta(\cos (\beta \pi / 2))^{\lceil 1 / \beta\rceil}} \frac{1}{k_c} e^{-\frac{1}{2} {k_c}^\beta \cos (\beta \pi / 2)},
\end{align}
where $c_\beta=2\pi e^{-2^\beta}$, and $U(T, k)=\mathcal{T} e^{-i \int_0^T[H_r(s)+k (H_i(s)-{E_i}_0(s))] \mathrm{d} s}$.

To bound the above truncation error by $\varepsilon$, make the right-hand side of the above equation to be less than or equal to $\epsilon$, then we can obtain
\begin{align}
k_c=\mathcal{O}\left(\left(\log \left(\frac{1}{\varepsilon}\right)\right)^{1 / \beta}\right),
\end{align}
which is seen as the query complexity of our algorithm when the Hamiltonian $\tilde{H}(t)$ is time-independent.

In addition, when the Hamiltonian $\tilde{H}$ is time-dependent, the evolution operator typically contains a time-ordering operator $\mathcal{T}$, necessitating a discretization of time, which inevitably introduces an additional error commonly referred to as the quadrature error. According to the conclusions of Lemma 10 in Ref.\cite{An2023a}, the quadrature error can be bounded as
\begin{align}\label{quadrature}
\left\|\int_{-k_c}^{k_c} g(k) U(T, k) \mathrm{d} k-\sum_{m=-k_c / h}^{k_c / h-1} \sum_{q=0}^{Q-1} c_{q, m} U\left(T, k_{q, m}\right)\right\| &\leqslant \frac{8}{3 C_\beta} k_c h^{2 Q}\left(\frac{e T \max _t\|\tilde{H}_i(t))\|}{2}\right)^{2 Q} \nonumber\\
&\leqslant \frac{8}{3 C_\beta} k_c h^{2 Q}\left(\frac{e T \max _t\|\tilde{H}(t)\|}{2}\right)^{2 Q}. 
\end{align}
In the above equation, the Gauss-Legendre quadrature is utilized, where $Q$ denotes the number of quadrature nodes. Here, $h$ is the step size used in the composite quadrature rule, and $h$ can be chosen to be a suitable value such that $k_c/h$ is an integer. And the $k_{q,m}$’s are the Gaussian nodes, $c_{q,m} = w_qg(k_{q,m})$, and the $w_q$’s are the Gaussian weights. It is worth emphasizing that both $k_{q,m}$ and $w_q$ can be considered as known quantities (by consulting tables), which are retrieved through the Gauss-Legendre quadrature methods.

By bounding the above quadrature error by $\varepsilon$, the parameter values can be obtained as
\begin{align}
h=\frac{1}{e T \max _t\|\tilde{H}(t)\|}, \quad Q=\left\lceil\frac{1}{\log 4} \log \left(\frac{8}{3 C_\beta} \frac{k_c}{\varepsilon}\right)\right\rceil=\mathcal{O}\left(\log \left(\frac{1}{\varepsilon}\right)\right),
\end{align}
so the overall number of unitaries $U(T, k_{q, m})$ in the Eq.\eqref{quadrature} is
\begin{align}
M=\frac{2 k_c Q}{h}=\mathcal{O}\left(T \max _t\|\tilde{H}(t)\|\left(\log \left(\frac{1}{\varepsilon}\right)\right)^{1+1 / \beta}\right),
\end{align}
which can be seen as the query complexity of our algorithm when the Hamiltonian $\tilde{H}(t)$ is time-dependent.

\subsubsection{Statistical error and the sampling complexity}
Next, we will discuss estimating the numerator and the denominator in Eq.\eqref{O_est} (the numerator is recorded as $N(O)$ here and hereafter, the denominator can be regarded as $N(I)$, but for the sake of easy distinction, we also denote it as $D$). Recalling that
\begin{align}\label{n_N_average}
\hat{N}(O)&=\|g\|_1^2\cdot\langle O(k',k)\cdot e^{i\theta(k',k)}\rangle, \nonumber\\
\overline{\hat{N}}(O)&=\frac{1}{n_N}\sum_{i=1}^{n_N}\hat{N}(O),
\end{align}
where $n_N$ denotes $n_N$ independent sampling processes with independent identical distribution $\mathrm{Pd}(\cdot)$ ($\mathrm{Pd}(k',k)=\mathrm{Pd}(k')\cdot\mathrm{Pd}(k)$), so it is obvious $\mathbb{E}{\overline{\hat{N}}(O)}=N(O)$. For convenience, we have omitted the parameters $k'$ and $k$. Then the estimator of $\langle O\rangle$ can be denoted by $\langle\hat{O}\rangle\equiv\overline{\hat{N}}(O)/\overline{\hat{D}}$.

$\hat{N}(O)$ and $\overline{\hat{N}}(O)$ are usually complex numbers, and their estimation needs to be divided into real part $\Re[\cdot]$ and imaginary part $\Im[\cdot]$ separately. For the estimation precision $\epsilon_N$ of $N(O)$ ($N(I)$), according to the Hoeffding inequality (see the Eq.(111) in Ref.\cite{Zeng2021}), and considering that $\Re\overline{\hat{N}}(O)$, $\Im\overline{\hat{N}}(O)$, $\in \left[-\|g\|_1^2\cdot\|O\|_{l_1},\quad \|g\|_1^2\cdot\|O\|_{l_1}\right]$, we know
\begin{align}
\mathrm{Pr}(|\Re\overline{\hat{N}}(O)-\mathbb{E}\Re\overline{\hat{N}}(O)|&\geqslant\epsilon_{N,\mathrm{Re}})\leqslant 2\exp{\left(-\frac{n_{N,\mathrm{Re}}\cdot\epsilon_{N,\mathrm{Re}}^2}{2\|g\|_1^4\cdot\|O\|_{l_1}^2}\right)}, \nonumber\\
\mathrm{Pr}(|\Im\overline{\hat{N}}(O)-\mathbb{E}\Im\overline{\hat{N}}(O)|&\geqslant\epsilon_{N,\mathrm{Im}})\leqslant 2\exp{\left(-\frac{n_{N,\mathrm{Im}}\cdot\epsilon_{N,\mathrm{Im}}^2}{2\|g\|_1^4\cdot\|O\|_{l_1}^2}\right)},
\end{align}
so the sample number of the numerator will be
\begin{align}\label{n_N}
n_{N,\rm{Re}}\geqslant \frac{K_{N,\mathrm{Re}}\cdot \|g\|_1^4\|O\|_{l_1}^2}{\epsilon_{N,\mathrm{Re}}^2},\nonumber\\
n_{N,\rm{Im}}\geqslant \frac{K_{N,\mathrm{Im}}\cdot \|g\|_1^4\|O\|_{l_1}^2}{\epsilon_{N,\mathrm{Im}}^2},
\end{align}
with a failure probability 
\begin{align}\label{delta_N}
\delta_{N,\mathrm{Re}}\leqslant 2\exp{\left(-\frac{K_{N,\mathrm{Re}}}{2}\right)},\nonumber\\
\delta_{N,\mathrm{Im}}\leqslant 2\exp{\left(-\frac{K_{N,\mathrm{Im}}}{2}\right)},
\end{align}
where $K_{N,\rm{Re}}>0$ and  $K_{N,\rm{Im}}>0$ represent the manually chosen parameter that regulates both the quantity of samples and the failure probability. Usually, we can set  $\epsilon_{N,\rm{Re}}=\epsilon_{N,\rm{Im}}=\epsilon_N/\sqrt{2}$, $K_{N,\rm{Re}}=K_{N,\rm{Im}}= K_N$, and $\delta_{N,\mathrm{Re}}=\delta_{N,\mathrm{Im}}=\delta_{N}/2$, then 
\begin{align}
\mathrm{Pr}(|\overline{\hat{N}}(O)-N(O)|\geqslant\epsilon_{N})\leqslant 1-(1-\delta_{N,\rm{Re}})(1-\delta_{N,\rm{Im}})\leqslant \delta_N.
\end{align}

Overall, taking into account the failure probabilities of numerator $N(O)$ estimation and denominator $D$ estimation, the estimated failure probability for $\langle O\rangle=N(O)/D(O)$ is
\begin{align}\label{delta}
\delta\leqslant1-(1-\delta_D)(1-\delta_N)\leqslant4\exp{(-\frac{K_D}{2})}+4\exp{(-\frac{K_N}{2})}.
\end{align}
Traditionally, we often take $K_D=K_N$, then $K_D=K_N\leqslant2\ln(8/\delta)$.

Following an analysis of the sampling complexity associated with both the numerator and the denominator, an estimation of the total error for $\langle O\rangle$ can be carried out,
\begin{align}
\left|\langle O\rangle-\langle\hat{O}\rangle\right| & =\left|\frac{N(O)}{D}-\frac{\overline{\hat{N}}(O)}{\overline{\hat{D}}}\right| \nonumber\\
& =\left|\frac{N(O) \overline{\hat{D}}-D \overline{\hat{N}}(O)}{D \hat{D}}\right| \nonumber\\
& =\left|\frac{N(O)\left(D+(\overline{\hat{D}}-D)\right)-D\left(N(O)-(\overline{\hat{N}}(O)-N(O))\right)}{D\left(D+(\overline{\hat{D}}-D)\right)}\right| \nonumber \\
& =\left|\frac{N(O) (\overline{\hat{D}}-D)+D \left(\overline{\hat{N}}(O)-N(O)\right)}{D\left(D+(\overline{\hat{D}}-D)\right)}\right| \nonumber\\
& =\left|\frac{\langle O\rangle (\overline{\hat{D}}-D)+ \left(\overline{\hat{N}}(O)-N(O)\right)}{D+(\overline{\hat{D}}-D)}\right| \nonumber\\
& \leqslant\frac{|\langle O\rangle|\epsilon_D+\epsilon_N }{D-\epsilon_D} \nonumber\\
&\leqslant\frac{(|\langle O\rangle|+1)\epsilon_D+\epsilon_N }{D} \nonumber\\
&\leqslant\frac{(\|O\|_{l_1}+1)\epsilon_D+\epsilon_N }{D},
\end{align}
where we have reasonably assumed that $\epsilon_D<D$. In practice, $\epsilon_D$ is usually taken as $\epsilon$, $\epsilon_N$ is usually taken as $\|O\|_{l_1}\epsilon$, so in order to make the above upper bound of the error smaller or equal to the permissible error $\eta$ synthesized by the numerator and the denominator, it should be satisfied that
\begin{align}\label{nDnN}
\epsilon\leqslant\frac{D}{(2\|O\|_{l_1}+1)}\eta\Longrightarrow \epsilon_D\leqslant\frac{D}{(2\|O\|_{l_1}+1)}\eta, \epsilon_N\leqslant\frac{\|O\|_{l_1}D}{(2\|O\|_{l_1}+1)}\eta.
\end{align}
Then according to Eq.\eqref{n_N} we know that $n_D$, $n_N$ should meet the following condition 
\begin{align}
n_{D,\rm{Re}}&\geqslant \frac{K_{D,\rm{Re}}\cdot \|g\|_1^4\|I\|_{l_1}^2 }{\epsilon_{D,\rm{Re}}^2}=\frac{K_{D,\rm{Re}}\cdot(2\|O\|_{l_1}+1)^2\|g\|_1^4}{2D^2\eta^2 }, \nonumber\\
n_{D,\rm{Im}}&\geqslant \frac{K_{D,\rm{Im}}\cdot \|g\|_1^4\|I\|_{l_1}^2 }{\epsilon_{D,\rm{Im}}^2}=\frac{K_{D,\rm{Im}}\cdot(2\|O\|_{l_1}+1)^2\|g\|_1^4}{2D^2\eta^2 }, \nonumber\\
n_{N,\rm{Re}}&\geqslant \frac{K_{N,\rm{Re}}\cdot \|g\|_1^4\|O\|_{l_1}^2}{\epsilon_{N,\rm{Re}}^2}=\frac{K_{N,\rm{Re}}\cdot(2\|O\|_{l_1}+1)^2\|g\|_1^4}{2D^2\eta^2 },  \nonumber\\
n_{N,\rm{Im}}&\geqslant \frac{K_{N,\rm{Im}}\cdot \|g\|_1^4\|O\|_{l_1}^2}{\epsilon_{N,\rm{Im}}^2}=\frac{K_{N,\rm{Im}}\cdot(2\|O\|_{l_1}+1)^2\|g\|_1^4}{2D^2\eta^2 }.
\end{align}
In addition, form the property of the norm we know that
\begin{align}
e^{-2\int_{0}^{T}({E_i}_{max}(s)-{E_i}_{0}(s))ds}\leqslant D=\langle\psi|\mathcal{\overline{T}} e^{-i\int_0^T \tilde{H}^\dagger(s) d s}\cdot I\cdot\mathcal{T} e^{-i\int_0^T \tilde{H}(s) d s}|\psi\rangle \leqslant e^{-2\int_{0}^{T}({E_i}_{min}(s)-{E_i}_{0}(s))ds}, \nonumber\\
\end{align}
where ${E_i}_{min(max)}(s)$ is the minimum (maximum) eigenvalue of $H_i(s)$, and in fact, ${E_i}_{min}(s)$ is the ground state eigenvalue of $H_i(s)$, and they can always be estimated within a reasonable interval by some algorithms \cite{Yuan2021,Huo2023,Zhang2024}. Denoting that $E_{avg}(T)=\int_{0}^{T}({E_i}_{max}(s)-{E_i}_{min}(s))ds/T\leqslant \Delta_{E_{i}}$, where $\Delta_{E_i}=\underset{s\in[0,T]}{\max}\{{E_i}_{max}(s)-{E_i}_{min}(s)\}$. It is worth noting that $E_{avg}(T)$ is a structural constant of the system $H(t)$, and its physical meaning is the average bandwidth of the anti-Hermitian part of the Hamiltonian $H(t)$, and $\Delta_{E_i}$ represents the maximum bandwidth of $H_i(t)$.

Considering ${E_i}_{min}(s)-{E_i}_{0}(s)\geqslant0$, and the failure probability given in Eq.\eqref{delta_N} and Eq.\eqref{delta}. Then the maximum sampling number can be taken as
\begin{align}\label{complexity_DN}
n_D&=\frac{2\ln(\frac{8}{\delta})\cdot(2\|O\|_{l_1}+1)^2\|g\|_1^4 e^{4T\Delta_{E_i} }}{\eta^2}=\mathcal{O}\left(\frac{\log(\frac{1}{\delta})\cdot(\|O\|_1+1)^2\|g\|_1^4}{\eta^2}\right), \nonumber\\
n_N&=\frac{2\ln(\frac{8}{\delta})\cdot(2\|O\|_{l_1}+1)^2\|g\|_1^4 e^{4T\Delta_{E_i} }}{\eta^2}=\mathcal{O}\left(\frac{\log(\frac{1}{\delta})\cdot(\|O\|_{l_1}+1)^2\|g\|_1^4}{\eta^2}\right), 
\end{align}
where we have assumed that the duration of the algorithm is $T=\mathcal{O}(1/\Delta_{E_i})$. The Eqs.\eqref{complexity_DN} are the sampling complexity when the Hamiltonian $H$ is time-dependent.

From Eqs.\eqref{complexity_DN}, it can be seen that the number of samples typically required increases exponentially with duration $T$, which is inherently an NP-hard problem. A similar issue arises in quantum imaginary-time evolution algorithms \cite{Huo2023}. Therefore, for general non-Hermitian systems $H$, we set the duration of our algorithm to $T=\mathcal{O}(1/\Delta_{E_i})$ to circumvent this problem. 

Specially, when $H$ is time-independent, 
\begin{align}
 D\geqslant e^{-2({E_i}_{min}-{E_i}_{0})T}\cdot p_g, \nonumber\\
\end{align}
where $p_g\equiv |\langle\phi_g|\psi\rangle|^2$ is the overlap between the ground state eigenvalue $|\phi_g\rangle$ of $H_i$ and the initial state $|\psi\rangle$,  ${E_i}_0$ can be chosen to ensure $({E_i}_{min}-{E_i}_{0})T=\log(\mathcal{O}(1))$, then the maximum sampling number can be bounded by
\begin{align}
n_D&=\mathcal{O}\left(\frac{\log(\frac{1}{\delta})\cdot(\|O\|_{l_1}+1)^2\|g\|_1^4}{\eta^2 p_g^2}\right), \nonumber\\
n_N&=\mathcal{O}\left(\frac{\log(\frac{1}{\delta})\cdot(\|O\|_{l_1}+1)^2\|g\|_1^4}{\eta^2 p_g^2}\right). 
\end{align}

\end{proof}

\section{\label{Hamiltonian_simulation}The subroutine for Hamiltonian simulation without time discretization error}
In this section, we review a Hamiltonian dynamics simulation method without discretization error proposed by Etienne Granet and Henrik Dreyer, and this method belongs to the category of Quantum Monte Carlo methods \cite{Granet2024}. The emphasis on this HSWDE method stems from three primary factors: Firstly, it is theoretically devoid of systematic errors resulting from time discretization. Specifically, when compared to other Quantum Monte Carlo algorithms (such as qDrift, qSwift, QCMC, etc), increasing the number of samples not only effectively mitigates statistical errors but also systematically reduces errors stemming from time discretization. Secondly, this method can suppress systematic errors caused by the finite resolution of quantum gates and is theoretically immune to such errors. Lastly, in practical scenarios involving noisy quantum computers, this method has demonstrated potential exponential advantages in terms of precision dependence. These characteristics will collectively render this algorithm highly competitive in the NISQ era. We will see the effect of using this HSWDE method in Fig.\ref{Loschmidt_amplitude_QMC}.

Given a Hermitian Hamiltonian $H(t)=\sum_{n} c_n \sigma_n$, where $\sigma_n=\{I,X,Y,Z\}^{\otimes n}$ is generalized pauli operators, and $\{c_n\}$ are real. For any rotation Pauli gate operation, it can always be regarded as the rotation Pauli gate operation with bigger angle, but with a certain probability and an attenuation (gain). If this principle is written in the form of an equation, that is
\begin{align}\label{principle}
(1+p)I+p\cdot e^{-i|c|\cdot\mathrm{sgn}(c)\sigma\cdot\tau}=\lambda e^{-ic\sigma\cdot d\tau'},
\end{align}    
where the angle $|c|\tau>|c|d\tau'>0$, $c\in\mathbb{R}$ and $\lambda$ is the attenuation coefficient. The purpose of setting the sign function 'sgn' on the left side of the equation above is to ensure that the probability $p$ is always positive. Solving the above Eq.\eqref{principle},
\begin{align}\label{principle_solution}
p&=\frac{\tan(|c|d\tau')}{\sin\tau+(1-\cos\tau)\tan{d\tau'} }\nonumber\\
&= \frac{|c|d\tau'}{\sin\tau}+\mathcal{O}({(d\tau')^2}),\quad |c|d\tau'\rightarrow 0^+,\nonumber\\
\lambda&=\frac{p\sin\tau}{\sin\tau'}\nonumber\\
&=\frac{1}{\cos{d\tau'}\cdot(1+\tan{|c|d\tau'}\cdot\tan(\frac{\tau}{2}))}\nonumber\\
&=1-d\tau'\cdot|c|\tan\frac{\tau}{2}+\mathcal{O}({(d\tau')^2}),\quad |c|d\tau'\rightarrow 0^+.
\end{align}
Recalling that the definition of Poisson process,
\begin{definition}
 A stochastic process $\{N (t), t \geqslant 0\}$ is called a \textbf{time homogeneous Poisson process} (referred to as a \textbf{Poisson process} for short) if it satisfies: \\
 (1) a counting process, and N (0)=0;\\
 (2) an independent incremental process, that is, for any $0<t_1<t_2<·····<t_n$, $N(t_1),N(t_2)-N(t_1),..,N(ta_n)-N(t_{n-1})$ are mutual independence; \\
 (3) Incremental smoothness, i.e. $\forall s, t\geqslant0,n\geqslant0$, the probability $p[N(s+t)-N(s)=n]=p[N(t)=n]$; \\ 
 (4) For any $t>0$ and sufficiently small $\Delta t>0$, $p[N(t+ \Delta t)-N(t)= 1]= \alpha\cdot \Delta t+o(\Delta t)$ ($\alpha >0$), and $P[N(t+ \Delta t)- N(t)\geqslant 2|=o(\Delta t)$,
\end{definition}
so the process given in the Eq.\eqref{principle} is actually a Poisson process, except for a coefficient $\lambda$. Then multiplying the two sides of the Eq.\eqref{principle} by $M$ times while setting $M\cdot d\tau'\rightarrow T$, and according to the Eq.\eqref{principle_solution}, then the Eq.\eqref{principle} becomes 
\begin{align}\label{principle_Nsteps}
[(1+p)I+p\cdot e^{-i|c|\cdot\mathrm{sgn}(c)\sigma\cdot\tau}]^M=[\lambda e^{-ic\sigma\cdot d\tau'}]^M \rightarrow \lambda_{tot} e^{-ic\sigma\cdot T},\quad \lambda_{tot}=e^{-T\sum_n\cdot|c_n|\tan(\tau_n/2)}.
\end{align}
The left side of the above Eq.\eqref{principle_Nsteps} meets the binomial distribution $B(M,p)$, and as we all know, when $M\rightarrow \infty$, while $p\rightarrow 0^+$, the binomial distribution will be transformed into Poisson distribution, i.e., $B(M,p)\rightarrow Poi(M\cdot p)=Poi(\frac{|c|T}{\sin\tau})$, where we have used the results of the Eq.\eqref{principle_solution}. Therefore, when we consider the Hamiltonian simulation task,
\begin{align}
e^{-iH\cdot T}=\lim_{d\tau' \to 0^+} (e^{-ic_1\sigma_1\cdot d\tau'}e^{-ic_2\sigma_2\cdot d\tau'}...e^{-ic_N\sigma_N\cdot d\tau'})^{T/d\tau'},
\end{align}  
and according to the principle we introduced, we know that each small angle rotating Pauli gates \{$e^{-ic_n\sigma_n\cdot d\tau'}$\} can be decomposed into a larger angle rotating Pauli gates \{$e^{-i|c_n|\cdot\mathrm{sgn}(c_n)\sigma_n \cdot d\tau'}$\} with multiplying the constant coefficient $1/\lambda_n$ (the total attenuation coefficient will be $\lambda_{tot}=e^{-T\cdot\sum|c_n|\tan(\tau_n/2)})$, and the number of times they appear in  time intervals $[0, T]$ follows the Poisson distribution $\mathrm{Poi}(\frac{|c_n|T}{\sin\tau_n})$. Furthermore, according to the property of Poisson distribution, we know that the occurrence time of these gates $\{e^{-i|c_n|\cdot\mathrm{sgn}(c_n)\sigma_n \cdot d\tau'}\}$ will follow the uniform distribution $U(0,T)$ between time intervals [0, T]. By arranging all the rotated Pauli gates generated through sampling in ascending order of time, the quantum gate sequence is generated, denoted as $V$. Since the generation of each type of gate is independent of each other, the total number of gates in $V$ satisfies the Poisson distribution $\mathrm{Poi}(T\sum_n \frac{|c_n|}{\sin\tau_n})$, and 
\begin{align}\label{envolution_total}
e^{-iHT}=\frac{1}{\lambda_{tot}}\cdot\mathbb{E}V=e^{T\cdot\sum_n|c_n|\tan\frac{\tau_n}{2}}\cdot\mathbb{E}V,
\end{align}
where the symbol '$\mathbb{E}$' represents the expectation, and $\lambda_{tot}$ has been defined in Eq.\eqref{principle_Nsteps}. 

It is worth noting that this theory is also applicable to the case where the system is time-dependent. All that is needed is to change all the output parameters related to time cumulants from products to integrals with respect to time, such as  $\mathrm{Poi}(T\sum_n \frac{|c_n|}{\sin\tau_n})\rightarrow  \mathrm{Poi}(\sum_n \int_{0}^{T}\frac{|c_n(s)|}{\sin\tau_n}\mathrm{d}s)$, $\lambda_{tot}=e^{-T\cdot\sum_n|c_n|\tan(\tau_n/2)}\rightarrow \lambda_{tot}=e^{-\sum_n\int_{0}^{T}|c_n(s)|\tan(\tau_n/2)\mathrm{d}s}$. Within this Hamiltonian simulation method, our scheme given in Eq.\eqref{O_est} further evolves as follows:
\begin{align}\label{O_est_HSWDE}
\langle O\rangle=&\frac{\langle\psi|\mathcal{\overline{T}} e^{-i\int_0^T \tilde{H}^\dagger(s) d s}\cdot O\cdot\mathcal{T} e^{-i\int_0^T \tilde{H}(s) d s}|\psi\rangle }{\langle\psi|\mathcal{\overline{T}} e^{-i\int_0^T \tilde{H}^\dagger(s) d s}\cdot I\cdot\mathcal{T} e^{-i\int_0^T \tilde{H}(s) d s}|\psi\rangle } \nonumber\\
=&\frac{\lambda_{tot}^{-2}\cdot\|g\|\cdot\mathbb{E}_{V',V}\mathbb{E}_{k',k,n}\langle {O_n}_{V',V}(k',k)\cdot e^{is(k',k)})\rangle}{\lambda_{tot}^{-2}\cdot\|g\|\cdot\mathbb{E}_{V',V}\mathbb{E}_{k',k}\langle I_{V',V}(k',k)\cdot e^{is(k',k)}\rangle} \nonumber\\
=&\frac{\mathbb{E}_{V',V}\mathbb{E}_{k',k,n}\langle {O_n}_{V',V}(k',k)\cdot e^{is(k',k)})\rangle}{\mathbb{E}_{V',V}\mathbb{E}_{k',k}\langle I_{V',V}(k',k)\cdot e^{is(k',k)}\rangle},
\end{align}
where we continue to derive along the conclusion in Eq.\eqref{O_est}, $\langle [\cdot](k',k)\rangle\equiv \langle\psi|\mathcal{\overline{T}} e^{i \int_0^T K_s'(k') d s'}[\cdot]\mathcal{T} e^{-i \int_0^T K_s(k) d s}|\psi\rangle=\lambda_{tot}^{-2}\cdot\mathbb{E}_{V',V}\langle\psi|[\cdot]_{V',V}|\psi\rangle$, $\langle\psi|[\cdot]_{V',V}|\psi\rangle\equiv\langle\psi|V'^\dagger[\cdot]V|\psi\rangle$, $\mathbb{E}_{V',V}$ represents the expectation with respect to $V'$ and $V$ operators generated by the method given above the Eq.\eqref{envolution_total}, and $\lambda_{tot}$ is also defined in Eq.\eqref{principle_Nsteps} and bellow Eq.\eqref{envolution_total}.

\section{\label{ND_measurements} The measurements of $N(O)$ and $D$}
Next, we given the quantum circuit diagram for implementing $N(O)$. The circuit diagram for measuring numerator estimator $\hat{N}(O)$ is shown in Fig.\ref{hybrid_algorithm}. In the Fig.\ref{hybrid_algorithm}, $U(T,k)$ and $U(T,k')$ are the subroutines of Hamiltonian simulation, which can be seen as black boxes, and can be implemented by any Hamiltonian simulation program. The rand number $k$, $k'$, and $O_n$ are generated by a sampling program according to the probability density function $\mathrm{Pd}(k)$ and the probability $p_n$ given bellow Eq.\eqref{O_est}, and the phase part of sampling $e^{is(k,k')}$ and $e^{is(n)}$ should be finally multiplied in the calculation result. When the initial state $|\psi\rangle$ passed through this quantum circuit, we have
\begin{align}
|\Psi_n\rangle=\frac{1}{\sqrt{2}}(|0\rangle\otimes O_n\cdot U(T,k)|\psi\rangle+|1\rangle\otimes U(T,k')|\psi\rangle).
\end{align}
Therefore, 
\begin{align}\label{XY_measurement}
\langle\Psi_n|X|\Psi_n\rangle&=\mathrm{Re}\left(\langle\psi|(U^\dagger(T,k')\cdot O_n \cdot U(T,k)|\psi\rangle \right)\nonumber\\
&=\mathrm{Re}\langle O_n(k,k')\rangle, \nonumber\\
\langle\Psi_n|Y|\Psi_n\rangle&=-\mathrm{Im}\left(\langle\psi|(U^\dagger(T,k')\cdot O_n \cdot U(T,k)|\psi\rangle \right)\nonumber\\
&=-\mathrm{Im}\langle O_n(k,k')\rangle. 
\end{align}
After that, the estimator $\hat{N}(O)$ given in Eq.\eqref{n_N_average} can be obtained by calculating 
\begin{align}\label{Nk1k2}
\hat{N}(O)=\|g\|_1^2\cdot\sum_n\left(\langle\Psi_n|X|\Psi_n\rangle-i\langle\Psi_n|Y|\Psi_n\rangle\right)\cdot e^{\theta(k,k')},
\end{align}
and then 
\begin{align}\label{N}
N(O)=\|g\|_1^2\cdot\mathbb{E}_{k',k}\langle O(k',k)\cdot e^{i\theta(k',k)}\rangle.
\end{align}

The algorithm for estimating $N(O)$ (or $D$) is given in Algorithm \ref{alg:ND}.

\section{\label{open_systems_vec} The isomorphic relation between open quantum systems and non-Hermitian systems with vectorization}
The Lindbald superoperator $\mathcal{L}$ in Eq.\eqref{master equation} will be mapped to
\begin{align}\label{Lindblad operator_matrix}
  \overline{\overline{\mathcal{L}}}(t)=-i[{H}(t)\otimes I-I\otimes H^{\mathrm{T}}(t)]+\sum_\mu{ [\Gamma_\mu\otimes({\Gamma^\dag_\mu})^\mathrm{T}-\frac{\Gamma^\dag_\mu\Gamma_\mu\otimes I}{2}-\frac{I\otimes({\Gamma^\dag_\mu\Gamma_\mu})^\mathrm{T}}{2} ] }.
\end{align}
After that the Eq.\eqref{master equation} will be mapped to
\begin{equation}\label{Lindblad equation vec}
  \frac{\mathrm{d}|\rho(t)\rangle\!\rangle}{\mathrm{d} t}=-iL(t)|\rho(t)\rangle\!\rangle,
\end{equation}
where we have rewritten $\overline{\overline{\mathcal{L}}}$ as $-iL(t)$ to make the above equation a Schr\"{o}dinger-like equation. Then we can get the solution of the above equation,
\begin{equation}\label{solution_master_equation}
  |\rho(t)\rangle\!\rangle=\mathcal{T}e^{-i\int_0^t L(\tau)\mathrm{d}\tau }|\rho(0)\rangle\!\rangle,
\end{equation}
where $|\rho(0)\rangle\!\rangle$ is the vectorization of the initial density operator $\rho(0)$. 

Meanwhile, in the context of vectorization, for the mean value of any observable $O$,
\begin{align}\label{O_OQS}
\langle O\rangle&=\mathrm{Tr}(O\rho(t))\nonumber\\
&=\langle\!\langle O^\dagger|\mathcal{T}e^{-i\int_0^t L(\tau)\mathrm{d}\tau }|\rho(0)\rangle\!\rangle \nonumber\\
&=(\||O^\dagger\rangle\!\rangle\|\cdot\||\rho(0)\rangle\!\rangle\|)\cdot\langle O^\dagger|\mathcal{T}e^{-i\int_0^t L(\tau)\mathrm{d}\tau }|\rho(0)\rangle  \nonumber\\
&=(\|O\|_F\cdot\|\rho(0)\|_F)\cdot\langle O^\dagger|\mathcal{T}e^{-i\int_0^t L(\tau)\mathrm{d}\tau }|\rho(0)\rangle.
\end{align} 
It is worth noting that when $O$ is a Hermitian operator, $O^\dagger$ in the above Eq.\eqref{O_OQS} can be replaced by $O$, at the same time, $O$ and $\rho(0)$, so $\langle\!\langle O^\dagger|$ and $|\rho(0)\rangle\!\rangle$, are given in the beginning. Meanwhile, according to Eq.\eqref{Cauthy_distribution_trans}, the evolution operator $\mathcal{T}e^{\int_0^t -iL(\tau)\mathrm{d}\tau }$ can be implemented, therefore, the problem of the dynamics simulation of open quantum systems may be addressed. Specially, in the similar way with the Eq.\eqref{HrHi}, $L(t)$ in the Eq.\eqref{Lindblad equation vec} can be written as
\begin{align}
L(t)=L_r(t)-iL_i(t),
\end{align}
where 
\begin{align}\label{LrLi}
L_r=&\frac{L+L^\dagger}{2}=H\otimes I-I\otimes H^\mathrm{T}+\frac{i}{2}\sum_{\mu}[\Gamma_\mu\otimes(\Gamma_\mu^\dagger)^\mathrm{T}-\Gamma_\mu^\dagger\otimes\Gamma_\mu^\mathrm{T}], \nonumber\\
L_i=&i\frac{L+L^\dagger}{2}=\frac{1}{2}\sum_\mu [\Gamma_\mu^\dagger\Gamma_\mu\otimes I+I\otimes (\Gamma_\mu^\dagger\Gamma_\mu)^\mathrm{T}-\Gamma_\mu\otimes(\Gamma^\dagger_\mu)^\mathrm{T}-\Gamma^\dagger_\mu\otimes\Gamma_\mu^\mathrm{T}]. 
\end{align}
We must emphasize that, in fact, most of the time $L_i$ does not satisfy the condition $L_i\geqslant0$. Therefore, we have to choose a constant $c_p$ such that $L_i + c_p\geqslant0$. Then  according to Eq.\eqref{Cauthy_distribution_trans}, Eq.\eqref{O_OQS} will become 
\begin{align}\label{O_OQS_comp}
\langle O\rangle=(\||O^\dagger\rangle\!\rangle\|\cdot\||\rho(0)\rangle\!\rangle\|)\cdot e^{c_pt}\cdot\langle O^\dagger|\int_\mathbb{R}g(k)\cdot\mathcal{T}e^{-i\int_0^t [L_r(\tau)+k\cdot (L_i(\tau)+c_p)]\mathrm{d}\tau }\mathrm{d}k|\rho(0)\rangle.
\end{align}
After that, the non-unitary evolution can be implemented on a quantum computers, and the sampling method is implemented according to the Eq.\eqref{O_est}. Comparing Eq.\eqref{O_OQS_comp} with Eq.\eqref{O_est}, we know that, compared to the simulation problem of the evolution of a typical non-Hermitian system, the simulation problem of an open quantum system after vectorization does not have the denominator part in Eq.\eqref{O_est}. This brings some simplification to the problem. However, it must be noted that this still can not completely solve the problem of exponentially growing variance after long-time evolution due to the presence of the compensation term $e^{c_pt}$ in the Eq.\eqref{O_OQS_comp}. This problem is NP-hard in typical imaginary-time evolution. In the following, we will prove that in an open quantum system, the variance can be bounded.

\subsection{\label{variance_bound}The variance bound}
Assuming that $L=\sum_n l_n \sigma_n$, then $\|L\|_{l_1}=\sum_n |l_n|$. After the evolution of time $T_f=\max_j\{-1/\mathrm{Re}\Delta_j\}$, where $\Delta_j$ is the eigenvalue of $\mathcal{L}$, the system is considered to have reached a steady state \cite{Wang2022}, and for open quantum systems, there is always $\mathrm{Re}\Delta_j<0$, except for the steady state, $\Delta_j=0$. It is worth noting that $\Delta_j$ can also be estimated by the classical or quantum algorithms Ref.\cite{Yuan2021,Zhang2024}. From Eq.\eqref{Lindblad operator_matrix}, it is evident that  
\begin{align}
c_p\leqslant \|L\|_{l_1}\leqslant2(\|H\|_{l_1}+\sum_\mu \|\Gamma_\mu\|_{l_1}^2), 
\end{align}
then the variance can be bounded by
\begin{align}
e^{c_p\cdot T_f}\leqslant e^{2(\|H\|_{l_1}+\sum_\mu \|\Gamma_\mu\|_{l_1}^2)\cdot T_f},
\end{align}
so the variance can be bounded by $e^{\mathcal{O}(1)}$, and the simulation can be implemented. In fact, the ground state eigenvalue of $L_i$ can be estimated, and may be much smaller than $\|L\|_{l_1}$. 

\subsection{\label{normal jump operators}Proof of $L_i\geqslant0$ for any normal jump operators}
As already proven previously, when the jump operators of the system are normal operators, the problem of exponentially growing variance can be completely avoided. Supposing that the jump operators of $\{\Gamma_\mu\}$ in Eq.\eqref{LrLi} are normal operators, and they have spectral decomposition $\Gamma_\mu=u_\mu\cdot d_\mu\cdot u_\mu^\dagger$, where $u_\mu$ is unitary matrix, and $d_\mu$ is the complex diagonal matrix. Then $L_i$ in Eq.\eqref{LrLi} has the following decomposition,
\begin{align}
L_i=&\frac{1}{2}\sum_\mu [\Gamma_\mu^\dagger\Gamma_\mu\otimes I+I\otimes (\Gamma_\mu^\dagger\Gamma_\mu)^\mathrm{T}-\Gamma_\mu\otimes(\Gamma^\dagger_\mu)^\mathrm{T}-\Gamma^\dagger_\mu\otimes\Gamma_\mu^\mathrm{T}] \nonumber\\
=&\frac{1}{2}\sum_{u}u_\mu\otimes u_\mu^*\cdot[d_\mu^*d_\mu\otimes I+I\otimes d_\mu d_\mu^*-d_\mu\otimes d_\mu^*-d_\mu^*\otimes d_\mu]\cdot(u_\mu\otimes u_\mu^*)^\dagger,
\end{align}
where $u_\mu\otimes u_\mu^*$  are obviously also unitary matrices, and the matrix between  $u_\mu\otimes u_\mu^*$ and  $(u_\mu\otimes u_\mu^*)^\dagger$ are obviously diagonal matrices, and we can record it as $\boldsymbol{d}_\mu$ for the convenience of description. After that, 
\begin{align}
(\boldsymbol{d}_\mu)_{m,n}=|d_m|^2+|d_n|^2-d_md_n^*-d_m^*d_n\geqslant0,
\end{align}
where we denote $(\cdot)_{m,n}\equiv[\cdot]_{mm}\otimes[\cdot]_{nn}$. Therefore, consider any vector $|v\rangle\!\rangle$ (state) with the same dimension, 
\begin{align}
\langle\!\langle v|L_i|v\rangle\!\rangle=\frac{1}{2}\sum_{\mu}\langle\!\langle v|\boldsymbol{u}_\mu\cdot\boldsymbol{d}_\mu\cdot\boldsymbol{u}^\dagger_\mu|v\rangle\!\rangle\geqslant0,
\end{align}
where we define $\boldsymbol{u}_\mu\equiv u_\mu\otimes u_\mu^*$. This conclusion means that $L_i$ is semi-positive definite (i.e., $L_i\geqslant0$) for normal jump operators $\{\Gamma_\mu\}$.

\end{document}